\documentclass[11pt]{article}

\usepackage[left=0.5in,right=0.5in,bottom=0.75in,top=0.75in]{geometry}
\usepackage{amsmath}
\usepackage{amssymb}
\usepackage{graphicx}
\usepackage{enumitem}
\usepackage{algorithm}
\usepackage[noend]{algpseudocode}
\graphicspath{{}}
\usepackage[space]{grffile}
\usepackage{tikz}
\usetikzlibrary{shapes.geometric, arrows}
\tikzstyle{arrow} = [thick,->,>=stealth]
\usepackage{natbib}
\PassOptionsToPackage{hyphens}{url}\usepackage{hyperref}
\hypersetup{colorlinks=true, citecolor=black, linkcolor=black, urlcolor=cyan}

\renewcommand{\abstract}[1]{
 \centerline{
 \begin{minipage}{0.7\linewidth}
 \hrule
 \vskip 0.1in
  \begin{center}
    {\bf Abstract}
  \end{center}
  #1
 \vskip 0.1in
 \hrule
 \end{minipage}}
 \vskip 0.3in}

\usepackage{amsthm}
\newtheorem{theorem}{Theorem}

\numberwithin{equation}{section}

\providecommand{\I}{\mathbf{I}}

\renewcommand{\u}{\mathbf{u}}

\providecommand{\W}{\mathbf{W}}
\providecommand{\x}{\mathbf{x}}
\providecommand{\X}{\mathbf{X}}
\providecommand{\y}{\mathbf{y}}
\providecommand{\bb}{\boldsymbol{\beta}}
\providecommand{\bh}{\hat{\beta}}
\providecommand{\bbh}{\hat{\boldsymbol{\beta}}}
\providecommand{\be}{\boldsymbol{\eta}}
\providecommand{\beh}{\hat{\boldsymbol{\eta}}}
\providecommand{\bg}{\boldsymbol{\gamma}}
\providecommand{\lam}{\lambda}

\providecommand{\FD}{\textrm{FD}}
\providecommand{\mFDR}{\textrm{mFDR}}
\providecommand{\zero}{\mathbf{0}}
\providecommand{\abs}[1]{\left\lvert#1\right\rvert}
\providecommand{\al}[2]{\begin{align}\label{#1}#2\end{align}}
\providecommand{\as}[1]{\begin{align*}#1\end{align*}}
\providecommand{\als}[2]{\begin{align}\label{#1}\begin{split}#2\end{split}\end{align}}
\providecommand{\inP}{\overset{p}{\longrightarrow}}
\providecommand{\inD}{\overset{d}{\longrightarrow}}
\newcommand\independent{\protect\mathpalette{\protect\independenT}{\perp}}
\def\independenT#1#2{\mathrel{\rlap{$#1#2$}\mkern2mu{#1#2}}}
\providecommand{\cor}{\textrm{Cor}}

\title{Marginal false discovery rate control for likelihood-based penalized regression models}
\author{Ryan Miller\\Department of Biostatistics\\University of Iowa
  \and
  Patrick Breheny\\Department of Biostatistics\\University of Iowa}
\date{\today}

\begin{document}

\maketitle

\abstract{The popularity of penalized regression in high-dimensional data analysis has led to a demand for new inferential tools for these models. False discovery rate control is widely used in high-dimensional hypothesis testing, but has only recently been considered in the context of penalized regression.  Almost all of this work, however, has focused on lasso-penalized linear regression. In this paper, we derive a general method for controlling the marginal false discovery rate that can be applied to any penalized likelihood-based model, such as logistic regression and Cox regression. Our approach is fast, flexible and can be used with a variety of penalty functions including lasso, elastic net, MCP, and MNet.  We derive theoretical results under which the proposed method is valid, and use simulation studies to demonstrate that the approach is reasonably robust, albeit slightly conservative, when these assumptions are violated.  Despite being conservative, we show that our method often offers more power to select causally important features than existing approaches.  Finally, the practical utility of the method is demonstrated on gene expression data sets with binary and time-to-event outcomes.}

\section{Introduction}

High-dimensional data poses a challenge to traditional likelihood-based modeling approaches.  Penalized regression, which can provide sparse models in which only a subset of the available features have non-zero coefficients, is an increasingly popular approach that is well suited to handle high-dimensional data.  Inferential methods for controlling the error rates of variable selection for these methods, however, have been limited, especially outside the least-squares setting.

In this manuscript, we build upon the recently proposed idea of marginal false discovery rates and extend them to the more general class of likelihood-based models, which includes generalized linear models such as logistic regression as well as Cox proportional hazards models.  Our presentation focuses mainly on the most popular penalized regression method, the least absolute shrinkage and selection operator, or lasso \citep{tibshirani_1996}, but the methods we develop apply to many other penalties, including the minimax concave penalty \citep[MCP; ][]{MCP}, the smoothly clipped absolute deviations \citep[SCAD; ][]{SCAD}, and the elastic net \citep{Elastic_Net}.

There is an enormous literature on false discovery rate control as it applies to separately testing each individual feature one-at-a-time, but the issue is much more complex in the regression setting.  A fundamental challenge of inference under variable selection is how to account for using the same data to both select features as well as to fit the model.  Nevertheless, several approaches to address this challenge have been proposed.

One approach, which we refer to as {\em sample splitting}, is based upon dividing the data into two parts, using the first part for variable selection and the second part for inference. \citet{Sample_Splitting} first proposed this approach using single split, and \citet{Meinshausen2009} extended it by considering multiple random splits.  More recently, \citet{CovTest} and \citet{Selective_Inference} have proposed a family of methods known as {\em selective inference}, which test the significance of each variable selection along the lasso solution path as $\lambda$ is decreased, conditional upon the other variables already active in the model. Based on this sequence of tests, formal stopping rules can be derived in order to control the false discovery rate at a specified level \citep{GSell2016}.  Several other approaches exist -- e.g., procedures based on the bootstrap \citep{Dezeure2017} or the idea of de-biasing \citep{Javanmard2014}.  A comprehensive review is beyond the scope of this manuscript, so we pay particular attention to sample splitting and selective inference as representative of two different frameworks, as described in Section~\ref{Sec:mfdr}.

All of these approaches, however, suffer from two primary drawbacks.  First, they are very restrictive, in the sense that they are often unable to select more than one or two features, even at high false discovery rates.  This is especially true in high dimensions.  Second, they are quite computationally intensive.  For both approaches, it typically takes several orders of magnitude longer to estimate the false discovery rate than to fit the model in first place.

To address both of these drawbacks, \citet{BrehenyMFDR} proposed controlling the marginal false discovery rate (mFDR), and showed that focusing only on this weaker definition of false discovery yields a method that is less restrictive and far less computationally intensive than approaches for controlling conditional FDRs.  \citet{BrehenyMFDR} considered only the case of linear regression models.  Here, we extend that work to the more general class of regression models based on likelihood; in particular, this class includes generalized linear models and Cox proportional hazards models.  We compare these newly developed methods to sample splitting and selective inference using simulated data and apply our approach to two case studies involving high-dimensional data, one with a time-to-event outcome and the other with a binary outcome, to demonstrate the practical utility of our method on real data.

\section{Marginal false discovery rates}
\label{Sec:mfdr}

False discoveries are straightforward to define when conducting single variable hypothesis tests; a false discovery occurs when a feature $X_j$ is declared to be associated with an outcome $Y$ even though the feature is actually independent of the outcome: $X_j \independent Y$. In the regression framework, where many variables are being considered simultaneously, the idea of a false discovery is more complicated. The most common approach, which we refer to as the \textit{fully conditional} perspective, is to consider a feature $X_j$ to be a false discovery if it is independent of the outcome conditional upon all other features: $X_j \independent Y | X_{k \neq j}$. Penalized likelihood methods typically result in only a subset of the available variables being active in the model, thereby motivating the \textit{pathwise conditional} perspective.  This perspective focuses on the model where $X_j$ first becomes active and conditions only on the other variables present in the model (denote this set $M_j$) at that time when assessing whether or not variable $j$ is a false discovery: $X_j \independent Y | X_k \text{ for } k \in M_j$.

\begin{center}
\begin{tikzpicture}[node distance=1cm]

\node(b)[text centered] {$B$};
\node(u)[below of = b, text centered] {$ $};
\node(a)[left of = u,  text centered, xshift = -1.5cm] {$A$};
\node(c)[right of = u, text centered, xshift = 1.5cm] {$C$};
\node(y)[below of = u, text centered] {$Y$};
 
\draw [arrow] (a) -- (b);
\draw [arrow] (a) -- (y);

\end{tikzpicture} \\
\end{center}

In this paper we derive our method under the less restrictive \textit{marginal false discovery} definition \citep{BrehenyMFDR}, which we illustrate using the causal diagram depicted above. In this diagram variable $A$ has a direct causal relationship with the outcome variable $Y$ and should never be considered a false discovery. Variable $C$ is independent of variable $Y$ regardless of any variables we adjust for and should always be considered a false discovery. 

Variable $B$, on the other hand, is a more subtle case: $B$ is is correlated with $Y$, but after adjusting for $A$, $B$ and $Y$ are independent. Depending upon the perspective taken, $B$ might or might not be considered a false discovery. In a fully conditional approach, $B$ is considered a false discovery.  In a pathwise approach, whether $B$ is a false discovery or not depends on whether $A$ is active in the model.

In practice, however, both of these approaches suffer from the difficulty of determining the $A-B-Y$ relationship: is $A$ driving changes in $Y$ and $B$ merely correlated, or vice versa?  Avoiding these complications is one of the primary motivations behind the marginal perspective, which is only concerned with false discoveries arising from variables like $C$.  Depending on the application, selecting variables like $B$ may or may not be problematic, but it is almost always the case that a pure noise variable like $C$ is the worst kind of feature to select.  This point, along with the fact that the marginal false discovery rate is easy to interpret, makes the mFDR broadly useful and informative, although certainly there are scenarios in which the conditional FDRs are useful as well.

\subsection{Penalized likelihood optimization}

Consider data of the usual form $(\y, \X)$, where $\y$ denotes the response for $i = \{1, \ldots, n\}$ independent observations, and $\X$ is a matrix containing the values of $j = \{1, \ldots, p\}$ explanatory variables such that entry $x_{i,j}$ corresponds to the value of the $j^{\textrm{th}}$ variable for the $i^{\textrm{th}}$ observation.  We assume the columns of $\X$ are standardized such that each variable has a mean of $0$ and $\sum_i \x_j^2 = n$.

The explanatory variables in $\X$ are related to $\y$ through a probability model involving coefficients $\bb$.  The fit of the model to the data can be summarized using the log-likelihood, which we denote $\ell(\bb|X,\y)$.  In the classical setting, $\bb$ is estimated by maximizing $l(\bb|X,\y)$.  However, this approach is unstable in high dimensions unless an appropriate penalty, denoted $P_{\lambda}(\bb)$, is imposed on the size of $\bb$.
In this case, $\bbh$ is found by minimizing the objective function
\al{eq:obj}{
  Q(\bb|X,\y) =  -\frac{1}{n} \ell(\bb|X,\y) + P_{\lambda} (\bb).
}

In the classical setting, the maximum likelihood estimate is found by setting the score, $\u(\bb) = \nabla \ell(\bb|X,\y)$, equal to zero.  The penalized maximum likelihood estimate, $\bbh$, is found similarly, although allowances must be made for the fact that the penalty function is typically not differentiable.  These penalized score equations are known as the Karush-Kuhn-Tucker (KKT) conditions in the convex optimization literature, and are both necessary and sufficient for a solution $\bbh$ to minimize $Q(\bb|X,\y)$.

In a likelihood-based regression model, the likelihood depends on $\X$ and $\bb$ through a linear predictor $\be = \X\bb$; in other words, we can equivalently express the likelihood in terms of a loss function $f(\be|\y)=-\ell(\bb|\X,\y)$.
In what follows, we assume that the loss function is strictly convex with respect to the linear predictors $\be$; note that this does not imply strict convexity with respect to $\bb$.
Under these conditions, any solution $\bbh$ that minimizes \eqref{eq:obj} with the lasso penalty $P_{\lambda} (\bb) = \lambda||\bb||_1$ must satisfy \citep{lasso_kkt}:
\begin{equation}
  \label{eq:kkt}
  \begin{alignedat}{2}
  \tfrac{1}{n}u_j(\bbh) &= \lam \textrm{ sign}(\hat{\beta}_j) &\quad \text{if } \hat{\beta}_j &\neq 0 \\
  \tfrac{1}{n}u_j(\bbh) &\in [-\lam,\lam]  &\quad \text{if }  \hat{\beta}_j &= 0
  \end{alignedat}
\end{equation}
for $j \in \{1, \ldots, p\}$.

\subsection{Marginal false discovery rate bounds for penalized likelihood methods}
\label{Sec:main-results}

In this section, we use classical distributional properties of the score function along with the KKT conditions given above to derive an upper bound for the number of marginal false discoveries in the lasso model.
The basic intuition behind the derivation is that, given certain regularity conditions, if feature $j$ is a marginally independent of $\y$, then $Pr(\bh_j \neq 0)$ is approximately equal to $Pr(\tfrac{1}{n}\abs{u_j(\bb)} > \lam)$, where the classical score function $\u$ is evaluated at the true value of $\bb$ provided that the log-likelihood is correctly specified (i.e., that the model assumptions hold).
Given this result, the asymptotic normality of the score allows us to estimate this tail probability, and with it, the expected number of marginal false discoveries at a given value of $\lam$.

Three regularity conditions are required for these theoretical results to hold.
These are given below, where we let $\W = \nabla^2f$ denote the $n \times n$
matrix of second derivatives of the loss function with respect to $\be$, such
that the classical Hessian matrix $\nabla^2 \ell (\bb) = -\X^T\W\X$.  We use
$\W$ to denote this matrix evaluated at the true value of $\bb$ and $\widehat{\W}$
if evaluated at the lasso estimate.  In addition, we let $v_j=\x_j^T\W\x_j$,
with $\hat{v}_j$ defined similarly.

\begin{description}[labelindent=0.5cm, leftmargin=*]
\item [(A1)] Asymptotic normality of the score function: $(\X^T\W\X)^{-1/2}\u(\bb) \inD N(\zero,  \I)$, where $\I$ denotes the $p \times p$ identity matrix.
\item [(A2)] Vanishing correlation: $\frac{1}{n}\x_j^T\W\X_{-j} \inP \zero$.
\item [(A3)] Estimation consistency: $\sqrt{n}(\bbh-  \bb)$ is bounded in probability.
\end{description}

(A1) is a standard result of classical likelihood theory and can be shown for many types of models.  (A3) is not a trivial condition, but has been studied and shown to hold for various models and various types of penalties under certain conditions.  (A2), on the other hand, is unlikely to be truly satisfied by most features in practice.
Our theoretical results illustrate what is required for the proposal to work perfectly in the sense of providing a consistent estimate for the mFDR.
In practice, (A2) serves as a worst-case scenario for the correlation structure in the sense that other correlation structures will, on average, lead to fewer false discoveries.
Viewed as an estimator, this means that the mFDR equation we will derive \eqref{eq:mfdr} is inherently conservative, in the sense that it will overestimate the true mFDR.
Viewed instead as a control procedure, this means that \eqref{eq:mfdr} provides a probabilistic upper bound on the false discovery rate, and that if we use this equation to limit the mFDR to, say, 10\%, we can be confident that the true mFDR is even smaller.

The fundamental reason for this is that if a pool of noise features are uncorrelated, a variable selection method may select several of them, whereas if they are correlated, the method will tend to select just one.  More explicit results are shown in Section~\ref{Sec:sim}, which shows that when (A2) is violated, the mFDR bound is less tight, as one would expect.
Assumptions (A1)-(A3) are further discussed in the specific cases of logistic and Cox regression later in this section.

We now formally state our main theoretical result.  In interpreting this result, it is important to keep in mind that only features that are marginally independent of the outcome will satisfy both assumption (A2) and $\beta_j=0$.
In other words, the theorem applies to variables like $C$ in the causal diagram of Section~\ref{Sec:mfdr}, but not variable $B$: although the regression coefficient for $B$ is zero, its correlation with $A$ violates (A2).

\begin{theorem}
  \label{Thm:main}
  For any solution $\bbh$ of the lasso-penalized objective \eqref{eq:obj}, we have $\bh_j \neq 0$ if and only if
  \al{eq:selection}{\tfrac{1}{n}\abs{u_j(\bbh) + v_j\bh_j} > \lam.}
  Furthermore, provided that that feature j satisfies (A1)-(A3) and $\beta_j=0$,
  \as{\frac{u_j(\bbh) + v_j \bh_j}{\sqrt{v_j}} \inD N(0, 1).}
\end{theorem}

\begin{proof}
  The first remark follows directly from the KKT conditions \eqref{eq:kkt} and the fact that, if the loss function $f(\be|\y)$ is strictly convex, then $\W$ is positive definite and $v_j$ is positive.  The asymptotic normality of the score function (A1) implies that the following Taylor series expansion holds:
\as{\u(\bbh) = \u(\bb) - \X^T\W\X(\bbh - \bb) + o_p(n)(\bbh-\bb).}
Since $\beta_j=0$, we then have
\as{u_j(\bbh) &= u_j(\bb) - \x_j^T\W\X_{-j}(\bbh_{-j} - \bb_{-j}) - \x_j^T\W\x_j\bh_j + o_p(n)\bh_j,\\
\intertext{or}
\tfrac{u_j(\bbh) + v_j \bh_j}{\sqrt{v_j}}  &= \tfrac{u_j(\bb)}{\sqrt{v_j}} - \sqrt{\tfrac{n}{v_j}}[\tfrac{1}{n}\x_j^T\W\X_{-j}][\sqrt{n}(\bbh_{-j} - \bb_{-j})] + o_p(1)\sqrt{\tfrac{n}{v_j}}\sqrt{n}\bh_j.}
Noting that $v_j/n=O(1)$, the first term on the right side of the equation converges to $N(0,1)$ by (A1); the second term goes to zero by conditions (A2) and (A3); and the third term also goes to zero by (A3).
\end{proof}

Theorem~\ref{Thm:main} therefore implies that the probability that feature $j$ is selected, given that it is marginally independent of the outcome, is approximately the probability that a random variable following a $N(0, v_j/n^2)$ distribution exceeds $\lam$ in absolute value.
In principle, the expected number of marginal false selections could be obtained by summing this probability over the set of marginally independent noise variables; in practice, since the identity of this set is unknown, a conservative alternative is to sum over all $p$ variables.
This leads to the following upper bounds for the number and rate of marginal false discoveries:
\als{eq:mfdr}{
  \widehat{\FD} &= 2 \sum_{j=1}^{p} \Phi \left(\frac{-n\lam}{\sqrt{\hat{v}_j}}\right)\\
  \widehat{\mFDR} &= \frac{\widehat{\FD}}{|S|},
}
where $S$ is the set of selected variables and $|S|$ its size.  Note that because $p$ is used as an upper bound for the total number of noise features, $\widehat{\FD}$ and $\widehat{\mFDR}$ will be somewhat conservative even when (A1)-(A3) are satisfied.  However, in the scenario where most features are null, as is often presumed in high dimensional applications, this effect is relatively minor.  For a given value of $\lambda$ the process of calculating $\widehat{\mFDR}$ is encapsulated by Algorithm~\ref{Alg:mfdr}.

\begin{algorithm*}
\caption{Calculating the mFDR upper bound}\label{Alg:mfdr}
\begin{algorithmic}[10]
\Procedure{}{}
	\State Estimate $\widehat{W} \gets \nabla^2f(\beh)$
	\For {\texttt{$j \in \{1, \ldots, p\}$}}
		\State $\hat{v}_j \gets \x_j^T\widehat{W}\x_j$
		\State $\widehat{\FD}_{j, \lam} = 2 \Phi \left(\frac{-n\lam}{\sqrt{\hat{v}_j}}\right)$  by the distributional result of Theorem 1
	\EndFor
	\State $\widehat{\FD}_\lam =  \sum_{j=1}^{p} \widehat{\FD}_{j, \lam}$
	\State $\widehat{\mFDR}_\lam = \text{min}\left(\frac{\widehat{\FD}_\lam}{|S_\lam|}, 1 \right)$
\EndProcedure
\State \textbf{return} $\widehat{\mFDR}_\lam $
\end{algorithmic}
\end{algorithm*}

It is worth noting that the version of Theorem~\ref{Thm:main} presented in this paper takes $p$ to be fixed.  It is of interest to extend Theorem~1 to allow $p \to \infty$ as $n \to \infty$ and determine how this affects convergence; we leave this as a topic to be pursued in a future manuscript.

\subsection{Penalty functions}

The form of the mFDR calculation \eqref{eq:mfdr} is determined by the KKT conditions \eqref{eq:kkt}.  Although the theorem in previous Section is specific to the lasso, many other penalties proposed in the literature have very similar KKT conditions. This in turn allows our results to be easily extended to other penalized methods.

For example, consider the elastic net \citep{Elastic_Net}, which utilizes two penalty parameters: $\lambda_1$ and $\lambda_2$, with $\lam_1, \lam_2 >0$. The elastic net solution is found by minimizing $-\frac{1}{n} \ell(\bb|X,\y) + \lambda_1 ||\bb||_1 + \frac{\lambda_2}{2}||\bb||_{2}^2 $. The resulting KKT conditions dictate that $\hat{\beta_j} \neq 0$ if and only if
\as{\tfrac{1}{n}\abs{u_j(\bbh) + v_j\bh_j} > \lam_1.}
Compared with the corresponding equation for the lasso \eqref{eq:selection}, the only change is that the right hand side of the selection condition has changed from $\lam$ to $\lam_1$.  Thus, equation \eqref{eq:mfdr} applies to the elastic net as well, with a similarly trivial change (replacing $\lam$ with $\lam_1$).  Note that the actual estimates $\bbh$ (and as a result $\widehat{\W}$, $\hat{v}$, and $\abs{S}$) may certainly change a great deal, and thus the resulting inferences may be very different, but the form of the mFDR upper bound is essentially identical.

Furthermore, in some cases, such as for MCP and SCAD, the form of the upper bound is exactly the same.  In other words, equation~\eqref{eq:selection} holds for these penalties as well as the lasso, and therefore the mFDR upper bound is unchanged from \eqref{eq:mfdr} -- although again, the actual estimates $\bbh$ and any quantities based on them will be very different.

It is worth noting here that MCP and SCAD differ from the the lasso by relaxing the degree of penalization on variables with large effects.  This leads to solutions with greater sparsity and reduced bias.  In particular, both theoretical analysis and simulation studies have demonstrated that convergence is typically faster for MCP and SCAD than for the lasso \citep{MCP,SCAD,Breheny2011}.  Thus, as we will see later in the paper, the accuracy of regularity condition (A3) is typically better at finite sample sizes for these estimators than for the lasso and as a consequence, the resulting mFDR bound is tighter (i.e., less conservative).

\subsection{Likelihood models}

The results derived in this manuscript are valid for any likelihood-based model for which classical asymptotic results hold, and are thus applicable to a wide variety of models.  In this section, we provide additional detail for two such models, logistic regression and Cox regression, that are of particular interest due to their widespread use.

\subsubsection{Logistic regression}

Suppose $y_i$ follows a Bernoulli distribution with success probability $\pi_i$. In logistic regression the logit of $\pi_i$ is modeled as a function of $\be = X\bb$.  This results in a likelihood consisting of the product of $n$ independent Bernoulli distributions with success probabilities $\pi_i = \exp(\x_i^T\bb)/\{1 + \exp(\x_i^T\bb)\}$ and $\W$ a diagonal matrix whose entries are given by $\pi_i(1-\pi_i)$.
The asymptotic normality of the score function for logistic regression (A1) can be found in, e.g., \citet{McCullagh1989}, while necessary conditions for $\sqrt{n}$-consistency of the lasso estimates for logistic regression (A3) are established in \citet{SCAD}.

\subsubsection{Cox regression}
\label{Sec:cox}

The results of Section~\ref{Sec:main-results} also apply to penalized Cox proportional hazards regression. Here the outcome of interest contains two components: a time, $y_i$, along with an accompanying indicator variable, $d_i$, where $d_i = 1$ indicates $y_i$ is an observed event time and $d_i = 0$ indicates $y_i$ is a right censoring time.

Let $t_1 < t_2 < \ldots < t_m$ be an increasing list of unique failure times indexed by $j$. The Cox model assumes a semi-parametric form of the hazard such that $h_i(t) = h_0(t)e^{\x_i^T \bb}$, where $h_i(t)$ is the hazard for observation $i$ at time $t$ and $h_0(t)$ is a common baseline hazard. Cox regression is based upon the partial likelihood \citep{Cox1972}
\begin{equation*}
L(\be)  = \prod_{j=1}^{m} \frac{\exp(\eta_j)}{\sum_{i \in R_j} \exp(\eta_i),} 
\end{equation*}
where $R_j$ denotes the set of observations still at risk at time $t_j$, known as the risk set.

Letting $\pi_{ij} = \exp(\eta_i)/\sum_{k \in R_j}\exp(\eta_k)$, the $i$th diagonal element of $\W$ is $\sum_j d_j\pi_{ij}(1-\pi_{ij})$, while its $i,k$th off-diagonal element is given by $-\sum_j d_j\pi_{ij}(1-\pi_{kj})$.  As discussed in \citet{Simon2011}, however, the off-diagonal elements of $\W$ are typically negligible except for very small sample sizes.  Thus, for the simulations we present in Section~\ref{Sec:sim}, we took the off-diagonal elements of $\W$ to be zero in order to speed up the calculations.  We found no meaningful difference in the calculation of $\widehat{\mFDR}$ when using the diagonal approximation of $\W$ in place of the full matrix.

The asymptotic normality of the score function for Cox regression (A1) is established in \citet{Andersen1982}, while necessary conditions for $\sqrt{n}$-consistency of the lasso estimates for Cox regression (A3) can be found in \citet{Fan_scad}.  There are, however, some additional considerations considering censoring with respect to assumption (A2) in Theorem~\ref{Thm:main}.  As discussed earlier, (A2) holds for features that are marginally independent of all other features as well as the outcome.  For (A2) to hold for Cox regression, however, the feature must also be independent of the censoring mechanism.  This additional requirement is needed because when a variable is related to the censoring mechanism, its distribution will drift over time as certain values are disproportionately removed from the risk set, which can induce correlations between variables that would otherwise be uncorrelated. The impact of this additional concern on $\widehat{\mFDR}$ is further assessed via simulation in Section 3.3. 

\section{Simulation studies}
\label{Sec:sim}

In this section we study the behavior our proposed method for controlling the mFDR via several simulation studies. In each study we generate $j \in \{1, \ldots, p\}$ features from standard normal distributions for $i = 1, \ldots, n$ subjects. We focus on penalized logistic and Cox regression models, although we also carry out some simulations involving normally distributed outcomes in order to compare our results with those in \citet{BrehenyMFDR}.

For logistic regression, binary outcomes are generated from independent Bernoulli distributions with parameter $\pi_i = \exp(\x_i^T \bb)/\{1 + \exp(\x_i^T \bb)\}$. For Cox regression scenarios, survival outcomes are generated from independent exponential distributions with rate parameter $\theta_i = \exp(\x_i^T \bb)$. And for linear regression scenarios outcomes are generated from the model $y_i = \x_i^T\bb + N(0,\sigma^2)$. Factors such as $p$, $n$, $\bb$, the correlation between features, and censoring are varied throughout these simulations.

In this section and the next, we compare our proposed mFDR control method with sample splitting and covariance testing.  The covariance testing approach was implemented using the {\tt covTest} package \citep{CovTest} and the \text{forwardStop} function of the \text{SelectiveInference} package \citep{Selective_Inference}; however, the current versions of these packages caution that this approach is considered ``developmental'' for logistic regression and is not currently implemented for Cox regression.  Although software does exist for sample splitting for linear regression via the {\tt hdi} package \cite{Dezeure2015}, the current version of the package does not offer methods for logistic or Cox regression models, so we manually implemented our own sample splitting approach.

\subsection{Accuracy, sample size, and correlation}
\label{Sec:accuracy}

In Section~\ref{Sec:main-results}, the mFDR bound \eqref{eq:mfdr} was shown to be tight in an asymptotic sense provided that condition (A2) was met.  To see how useful this result is in practice, we wish to observe how tight this bound is at finite sample sizes as well as when condition (A2) is violated -- i.e., when noise features are correlated.
We evaluate the accuracy of the mFDR bound by comparing $\widehat{\mFDR}$ with the true mFDR -- i.e., the empirical proportion of noise features selected -- at each value along a fixed $\lambda$ sequence, averaged across 1,000 simulated data sets.
In this simulation, time-to-event outcomes were uncensored; the effect of censoring on $\widehat{\mFDR}$ is considered in the next section.

We generate our data using $p = 100$ with $\beta_{1:4} = 10/\sqrt{n}$ and $\beta_{5:40} = 0$ while varying $n$ from 200 to 1,600 in increments of 200. We take $\bb$ to be function of $n$ in order to maintain the difficulty of the variable selection problem as $n$ increases.
Without this provision, the feature selection problem becomes too easy at large sample sizes and the ``interesting'' region where the mFDR is not close to 0 or to 1 is very small.

In total we assess twelve different scenarios consisting of each possible combination of:
\begin{itemize}[leftmargin=*, labelindent=0.5cm]
\item Two different correlation structures: independent noise variables and correlated noise variables
\item Three different regression methods: linear regression, logistic regression, and Cox regression
\item Two different penalties: lasso and MCP. 
\end{itemize}
For our correlated setting noise variables are given an autoregressive correlation structure based upon their index such that $\cor(\x_j,\x_k) = 0.8^{|j - k|}$. We display the mean value of $\widehat{\mFDR}$ at the $\lambda$ value where the average observed proportion of noise variables is 20\%.

\begin{figure} [!htb]
 \centering
  \includegraphics[width=.9\textwidth]{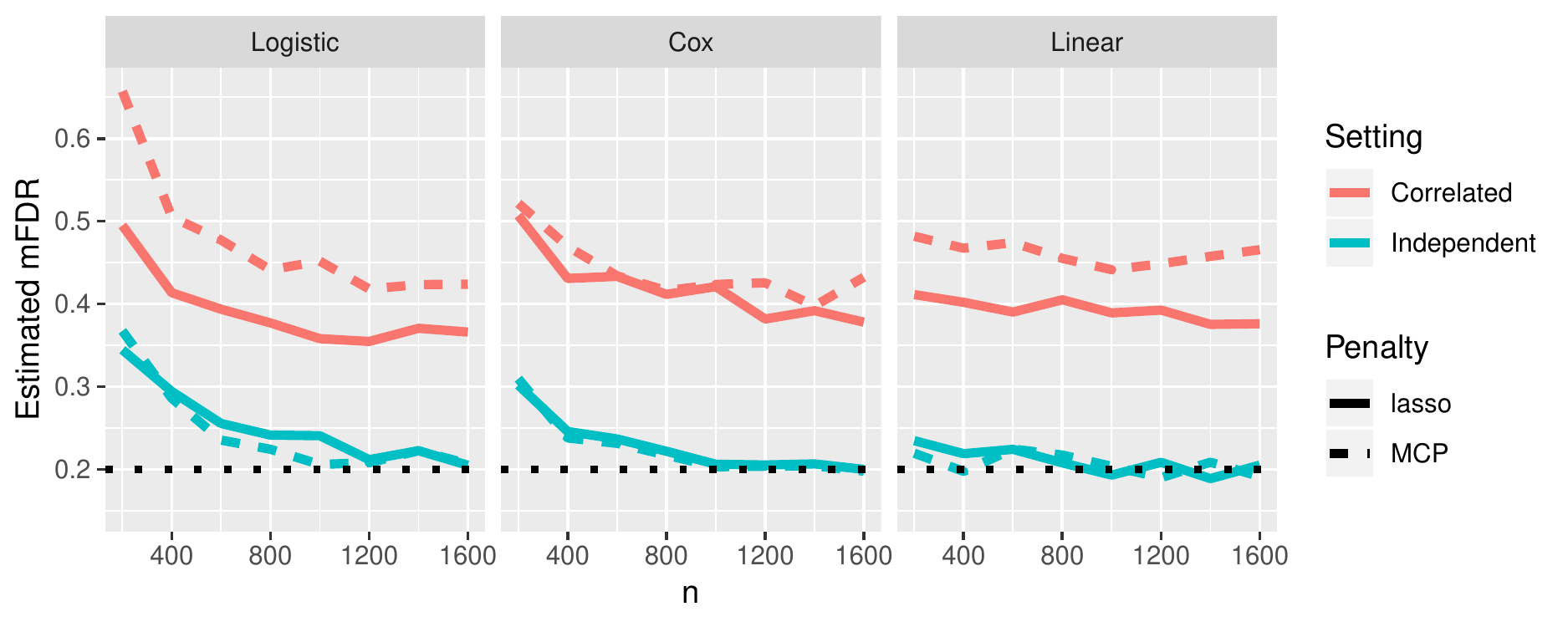}
  \caption{\label{Fig:converge}Finite-sample accuracy of the mFDR bound \eqref{eq:mfdr}.}
\end{figure}

Focusing first on the case of independent noise features, we see that unlike in the linear regression case, where $\widehat{\mFDR}$ is essentially perfect at all sample sizes considered, for logistic and Cox regression $\widehat{\mFDR}$ tends to provide a conservative upper bound at small sample sizes, although this effect does diminish and the bound becomes tighter as $n$ increases.
Furthermore, the MCP penalty leads to slightly faster convergence and tighter bounds than the lasso penalty; this is particularly noticeable in the logistic and Cox regression settings.

When noise features are correlated, the $\widehat{\mFDR}$ bound is always conservative, at all sample sizes.  In terms of Theorem~\ref{Thm:main}, when noise features are correlated, (A2) is violated because the term $\frac{1}{n}\x_j^T\W\X_{-j}$ does not converge to zero, leaving behind a remainder term that is unaccounted for in \eqref{eq:mfdr}.
Nevertheless, while this does mean that the bound is less tight, it does not prevent the use of the method for the purposes of mFDR control.
Furthermore, in this case, the effect is relatively slight: for the most part, $\widehat{\mFDR}$ is between 35\% and 40\% when the true mFDR is 20\%.

\subsection{Features associated with censoring}

As discussed in Section~\ref{Sec:cox}, the presence of censoring in Cox regression means that condition (A2) will also be violated if noise features are associated with the censoring mechanism.  Here, we assess the impact such an association has on the $\widehat{mFDR}$.

Consider two scenarios, A and B. In each scenario censoring times are generated from exponential distributions with $\theta_{i} = \exp(\x_i^T \bg)$, with
\as{\bg_A &= (0, 0, 0, 0, .45, -.45, .45, -.45, 0, \ldots, 0)}
for scenario A, while in scenario B all variables are independent of censoring (i.e., $\bg_B = \zero$).  By design, each of these scenarios results, on average, in a censoring rate of 50\%.  The same set of $n = 500$ true failure times is used for both scenarios and are generated from independent exponential distributions with rate parameter $\theta_i = \exp(\x_i^T \bb)$, where $\beta_{1:4} = .45$ and $\beta_{5:40} = 0$.  The value 0.45, used in $\bg$ and $\bb$, is approximately $10/\sqrt{n}$ for this scenario, so that the results of this simulation are comparable with those displayed in Figure~\ref{Fig:converge}.



In scenario A, several noise features are associated with the censoring mechanism.  Here, when $\widehat{\mFDR}$ is 10\%, the observed false discovery rate averaged across 1,000 simulations is 4.03\% using the lasso penalty and 6.31\% using MCP.  In scenario B, where all censoring assumptions are met, the observed false discovery rates are 5.01\% and 6.85\% for the lasso and MCP, respectively.  These results demonstrate that noise variable associated with the censoring mechanism will lead to a more conservative upper bound; however, the effect is negligible compared to the considerations described in Section~\ref{Sec:accuracy}.

\subsection{Comparison with other methods}
\label{Sec:sim-comp}

In this section, we compare our proposed mFDR approach with other methods of variable selection used in the high-dimensional setting. We generate our data based upon the structure depicted in the causal diagram of Section~\ref{Sec:mfdr}, using $n = 400$ and $p = 1000$ with ten variables causally related to the outcome (like ``A''), such that $\beta_{1:10} = b$ and vary $b$ throughout the study. Corresponding to each of these variables, we generate nine correlated ($\rho = 0.5$) variables (akin to ``B''). The remaining 900 variables are noise (akin to ``C''), however they are correlated with each other such that $\cor(\x_j,\x_k) = 0.8^{|j - k|}$, thereby creating a situation where the mFDR upper bound will be conservative.

\begin{figure} [htb!]
 \centering
  \includegraphics[width=0.75\textwidth]{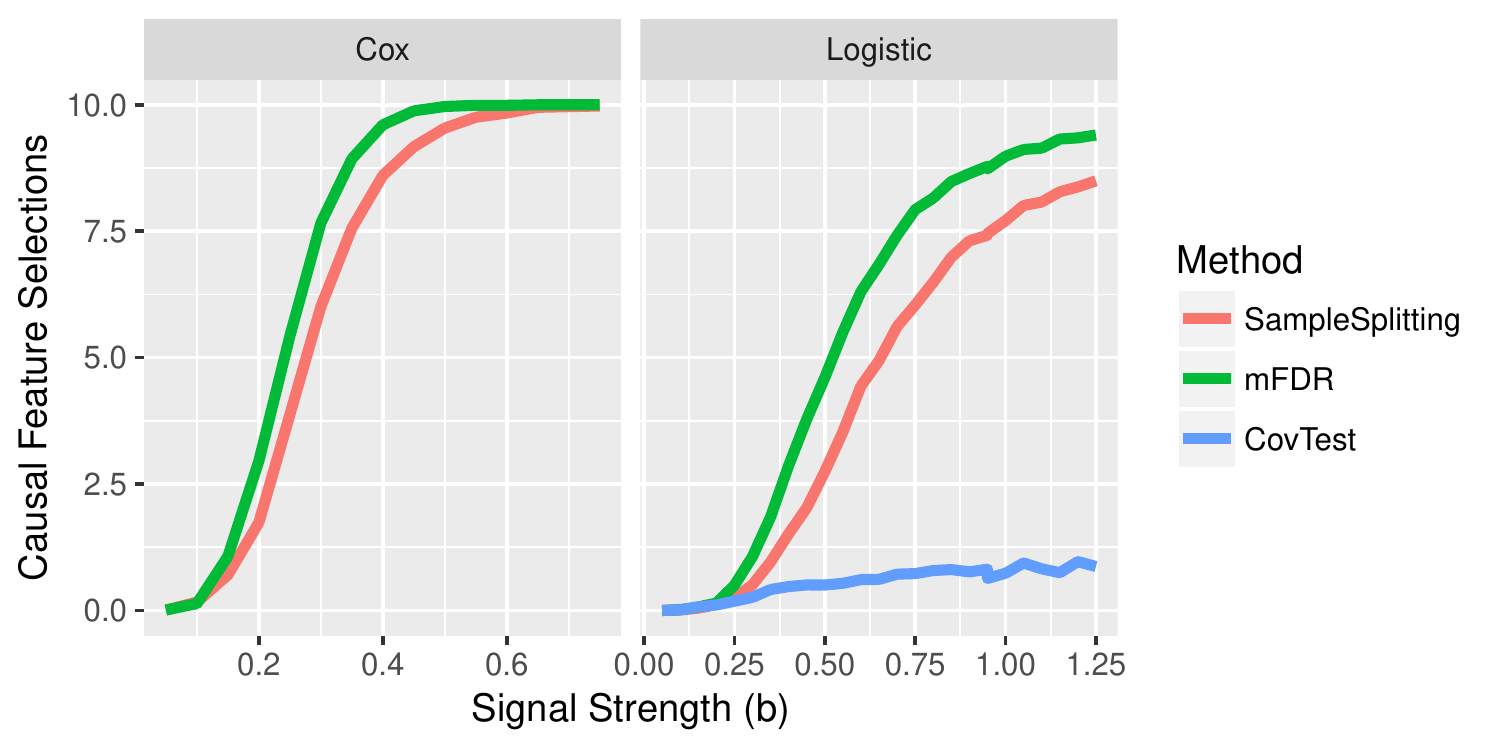}
  \caption{\label{Fig:lassopower} The number of true, ``A", variables selected by each lasso based method of false discovery rate control, averaged across 1,000 simulation iterations plotted as a function of $\beta$.}
\end{figure}


For each of the following methods, we assess the average number of selections for each variable type:
\begin{itemize}[labelindent=0.5cm, leftmargin=*]
\item Our mFDR approach, where variables are selected by the lasso using the smallest $\lambda$ with $\widehat{\mFDR} \leq .1$.
\item Sample splitting \citep{Sample_Splitting}, where we first fit a lasso regression model on half of the data to select the top 20 variables. We then use the remaining data to fit an unpenalized regression model on the variables selected in the first stage. With the unpenalized model we conduct traditional Wald hypothesis tests on the regression coefficients and apply the Benjamini-Hochberg procedure \citep{BH_1995} to control the false discovery rate at $10\%$. Here, we limit the first-stage selections to 20 variables so that in the second stage, the model contains 10 events per variable \citep{peduzzi_epv}.
\item Covariance testing \citep{CovTest}, which we use in conjunction with the forward stopping rules proposed by \citet{GSell2016} to control the pathwise false discovery rate at $10\%$. 
\item Univariate testing, where we fit separate, unpenalized regression models to each variable individually, and then adjust the resulting p-values using the Benjamini-Hochberg procedure to control the false discovery rate at $10\%$.
\item Cross-validation (CV), where 10-fold CV is used to choose $\lam$ for the lasso model; note that this approach makes no attempt to control the false discovery rate.
\end{itemize}

The comparison of FDR control procedures for the lasso (mFDR, sample splitting, and covariance testing) is shown in Figure~\ref{Fig:lassopower}.  Compared to other false discovery rate control methods for penalized regression, the mFDR approach is more powerful -- at the same nominal FDR, the marginal approach selects more causally important variables at each signal strength.  This is not a surprising result: a primary motivation of the marginal approach is that, by adopting a less restrictive definition of false discovery, we could improve power.  The covariance testing approach in particular is particularly conservative in the logistic regression setting, selecting far fewer causal features than either mFDR or sample splitting.  As of the time of this writing, the covariance test has not been extended to Cox regression and is therefore not shown in the left panel of Figure~\ref{Fig:lassopower}.

\begin{figure} [!ht]
 \centering
  \includegraphics[width=0.75\textwidth]{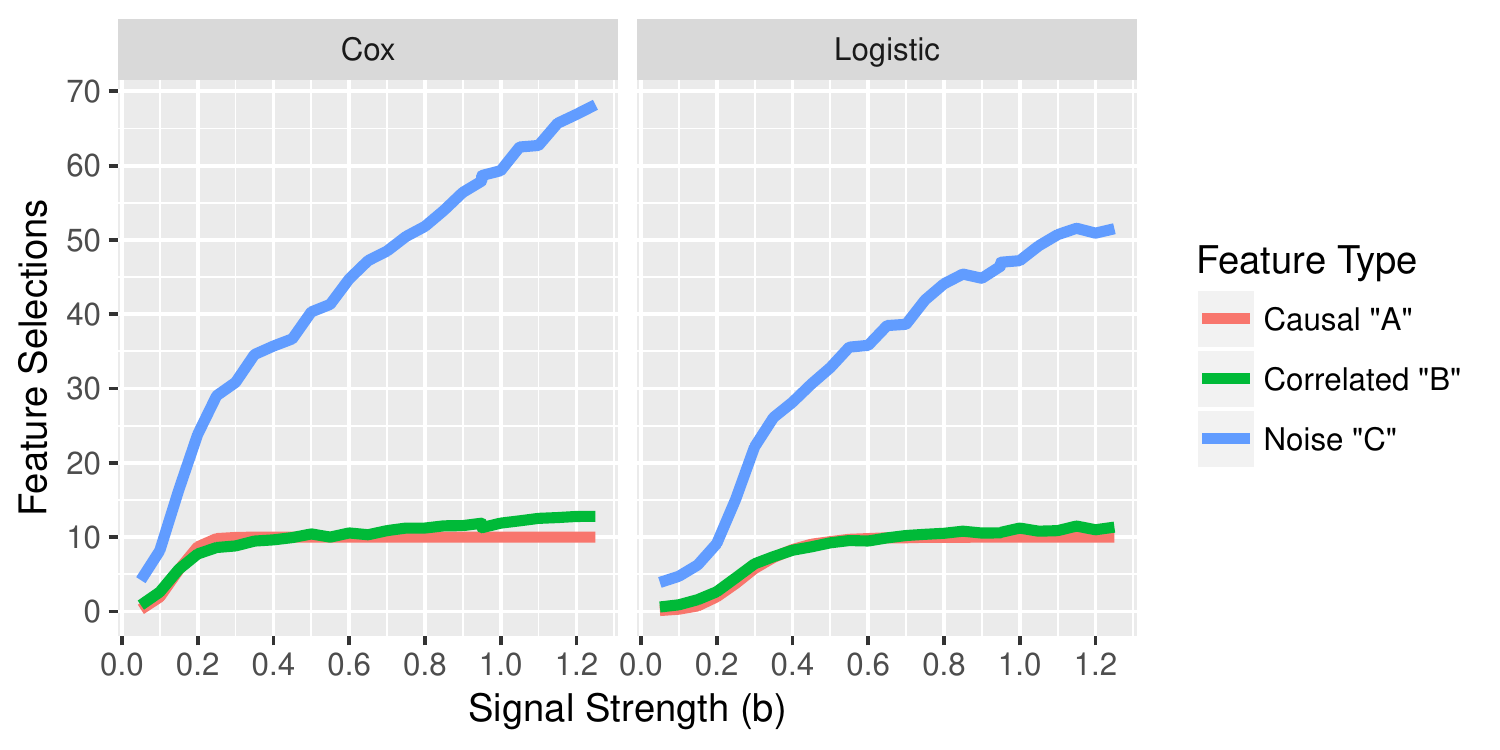}
  \caption{\label{Fig:cv} The average number of selections for each type of variable depicted in the causal diagram of Section~\ref{Sec:mfdr}, when cross-validation is used to select $\lam$.}
\end{figure}

Simulation results for cross validation are shown in Figure~\ref{Fig:cv}.  As the figure shows, while cross-validation has excellent sensitivity, selecting nearly all the causal (``A'') variables even at low signal strength, it does not control the number of false selections.  In fact, at all levels of $b$, the majority of the features in the CV-selected model are noise features.  This problem becomes increasingly worse as the signal strength increases.  For example, at $b=0.8$, over 70\% of the features in the CV model are noise features.  In contrast, as may be seen in Figure~\ref{Fig:univariate}, the vast majority of the features selected by the mFDR control approach are causally related to the outcome.

Figure~\ref{Fig:univariate} compares the mFDR and univariate approaches, both of which are designed to control the proportion of marginal false discoveries without making any claims regarding features indirectly associated with the outcome. The top panels show that the two methods are comparable in terms of their ability to select the true causal variables.  The advantage of using a regression-based approach over univariate testing, however, is clearly seen when looking at the number of other selections.  Lasso with mFDR control greatly reduces the number of correlated (``B'') selections compared to univariate testing: whereas univariate testing tends to identify dozens of indirectly associated features as discoveries, lasso with mFDR control tends to select at most two.

The mFDR control approach also selects far fewer noise (``C'') variables than univariate testing.  The reason for this is that the large number of ``B'' features selected by univariate testing allow it to select a much larger number of features while maintaining the overall FDR at 10\%.

\begin{figure} [!htb]
 \centering
  \includegraphics[width=0.8\textwidth]{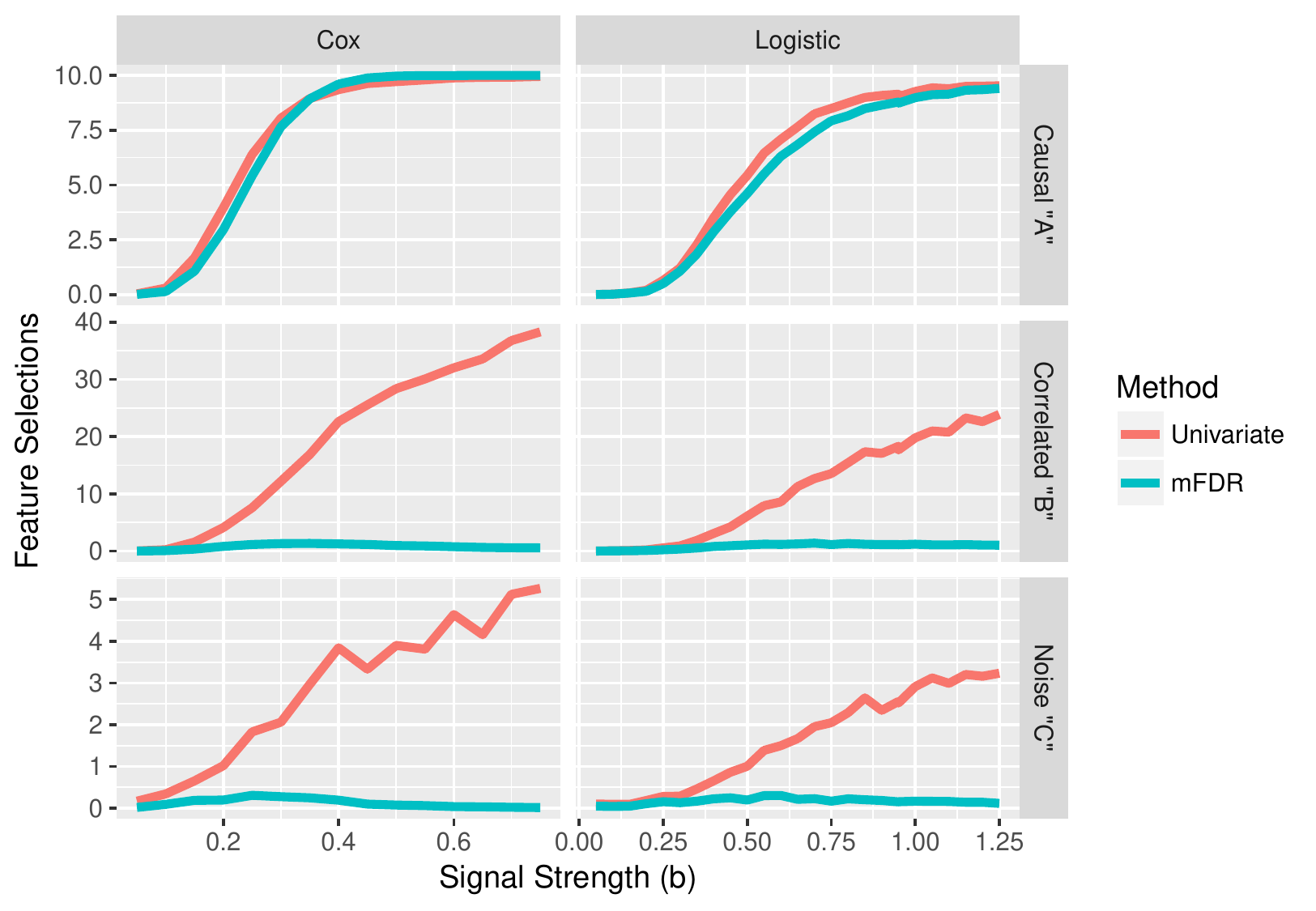}
  \caption{\label{Fig:univariate} The average number of selections, for each type of variable depicted in the causal diagram of Section~\ref{Sec:mfdr}, is plotted as function of $\beta$ for each method which controls the false discovery rate.}
\end{figure}

For this reason, Figure~\ref{Fig:univariate} is potentially somewhat misleading in terms of comparing mFDR and univariate testing in terms of their power to identify type ``A'' features.
Table~\ref{Tab:univariate} presents an alternative metric for comparing the univariate and mFDR approaches.  Given the large difference between the two approaches with respect to number of indirectly associated (``B'') features selected, the table presents the ratio of the number of causally associated features to the number of noise features -- i.e., the A:C ratio -- for various values of $\beta$.  Whereas Figure~\ref{Fig:univariate} suggests that the two approaches were approximately equally effective at identifying the truly important variables, Table~\ref{Tab:univariate} shows that the lasso-mFDR approach is far more powerful in terms of the number of true discoveries per noise feature selected.

\begin{table}[ht!]
\centering
\caption{\label{Tab:univariate} The ratio of Causal:Noise (A:C) feature selections for various values of the signal strength $b$.}
\begin{tabular}{c | r r r r}
  \hline
  & \multicolumn{2}{c}{Cox} & \multicolumn{2}{c}{Logistic}\\
 $b$ & Univariate & Lasso-mFDR & Univariate & Lasso-mFDR \\ 
  \hline
  0.25 & 3.5 & 17.7 & 2.3 & 3.3 \\ 
  0.35 & 3.0 & 36.2 & 4.9 & 11.1 \\ 
  0.45 & 2.9 & 102.2 & 5.3 & 15.3 \\ 
  0.65 & 2.4 & 333.3 & 4.6 & 32.6 \\ 
  0.75 & 1.9 & 750.0 & 4.1 & 47.6 \\ 
   \hline
\end{tabular}
\end{table}

Indeed, for the univariate approach, the rate of true variable selections per noise variable selection is not only low, but in fact decreases as the signal strength increases. In contrast, with the mFDR approach, the causal:noise selection ratio increases with signal strength, as one would hope.

\section{Case studies}

\subsection{Lung cancer survival and gene expression}

\citet{Shedden2008} studied the survival of 442 early-stage lung cancer subjects. Researchers collected high-dimensional gene expression data on 22,283 genes and additional clinical covariates of age, race, gender, smoking history, cancer grade, and whether or not the subject received adjuvant chemotherapy.  The retrospective study used time to death as its outcome, which was observed in 236 (53.3\%) of the subjects, the remainder being right censored.

\begin{figure} [!htb]
 \centering
  \includegraphics[width=0.65\textwidth]{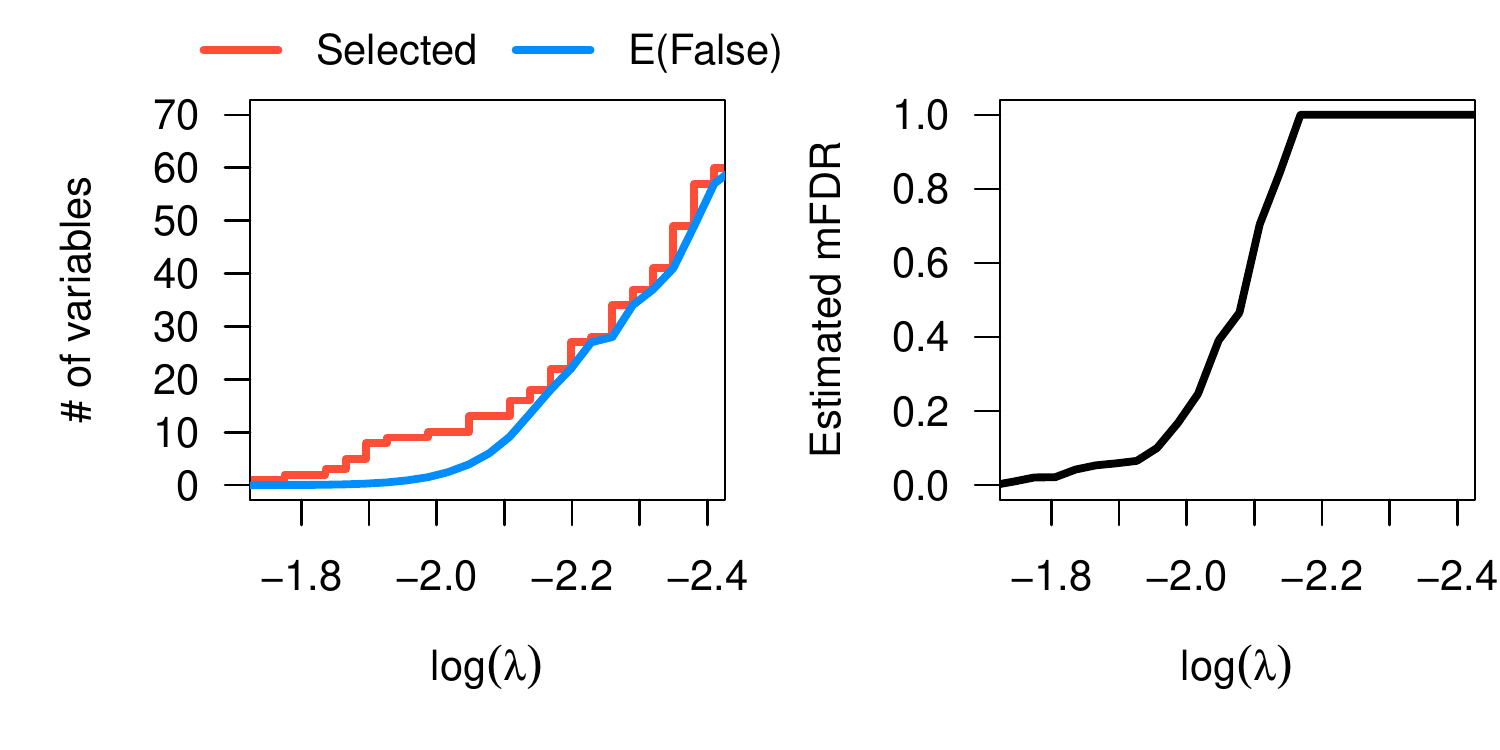}
  \caption{\label{Fig:Shedden} Marginal false discovery estimates for the Shedden data.  (left) Total number of genes selected and expected number of marginal false discoveries; (right) Estimated mFDR.}
\end{figure}

In our analysis we aim to select important genes while controlling the marginal false discovery rate.  We use a semi-penalized Cox proportional hazards regression model that allows all of the clinical covariates to enter the model unpenalized while using the sparsity introduced by penalization to screen for additional genes related to survival.  We aim to limit the mFDR to 10\% and compare our results to those of other methods.  Note that in this example, the number of features $p$ in equation~\eqref{eq:mfdr} is the number of penalized features (i.e., the number of genes), since the clinical covariates will always be in the model.

Figure~\ref{Fig:Shedden} illustrates the mFDR estimates for this study.  On the left, we see a gap between the number of selected genes and the expected number of false selections at $\log(\lam) \approx -2$.  This indicates that for these values of $\lam$, many of the genes selected by the lasso are likely to be truly related to survival, as it would be unlikely for so many noise features to be selected merely by random chance.  On the right, these expected number of false discoveries is plotted as a fraction of the selected features (i.e., the mFDR).

\begin{table}[htb!]
\centering
\caption{\label{Tab:case-studies} Results for different choices of $\lambda$ in applying the lasso and MCP penalties to the Shedden and Spira data.  mFDR = Marginal false discovery rate (\%). CVE = Cross-validation error.  MCE = Misclassification error (\%).}
\vspace{.2cm}
\begin{tabular}{l l r r r r r @{\hskip 0.5in} r r r r r r}
  \hline
  & & \multicolumn{5}{c}{Shedden case study} & \multicolumn{6}{c}{Spira case study}\\
  & & $\lambda$ & EF & S & mFDR & CVE & $\lambda$ & EF & S & mFDR & CVE & MCE\\
  \hline
  lasso & CV & 0.118 & 13.53 & 16 & 84.6 & 11.05 & 0.029 & 71 & 71 & 100 & 1.00 & 24.5  \\
  & CV(1se) & 0.180  &  0 &  0 &  -  & 11.17 & 0.063 & 32 & 32 & 100 & 1.07 & 25.0  \\
  & mFDR & 0.146 & 0.52 & 8 & 6.6 & 11.10 & 0.146 & 0.78 & 10 & 7.8 & 1.29 & 30.7 \\
  MCP & CV & 0.155 & 0.16 & 8 & 8.2 & 11.21 & 0.067 & 14 & 14 & 100 & 1.13 & 29.7 \\
  & CV(1se) & 0.180 & 0 & 0 & - & 11.22 & 0.093 & 10 & 10 & 100 & 1.19 & 31.8  \\
  & mFDR & 0.155 & 0.16 & 2 & 8.2 & 11.21 & 0.151 & 0.42 & 5 & 8.4 & 1.29 & 35.9 \\
  \hline
\end{tabular}
\end{table}

Table~\ref{Tab:case-studies} displays results for three potential approaches to choosing $\lam$: the value that minimizes cross-validation error (CV), the value that comes within 1 standard error (SE) of minimizing CV (1se), and the smallest value of $\lam$ satisfying $\mFDR < 10\%$.
From the table, we can see that there appears to be an inherent tradeoff between the optimal prediction and restricting a model to containing few noise features.
In this case, we see that the lasso model with lowest prediction error has an estimated mFDR of 84.6\%, indicating that despite its prediction accuracy, a substantial fraction of the 16 features selected by the model may be noise. Conversely, if we aim to control the mFDR at 10\%, we select only 8 features.  Although we can be more confident that these 8 features are truly related to survival, the model is somewhat less accurate from a prediction standpoint (CV error of 1313 compared to 1308).  The CV (1se) choice $\lam$ aims to select a parsimonious model whose accuracy is comparable with the best model \citep{Hastie2009}. However, in this example the CV (1se) method selects the null model, giving the impression that no genes can be reliably selected.  Note that the $\mFDR < 10\%$ approach is also within 1 SE of the best model, and therefore may also be thought of as striking an appropriate balance between prediction and parsimony in this example.

An alternative to lasso-penalized regression is to use the MCP penalty, which relaxes the degree of shrinkage for large coefficients, thereby allowing a smaller number of features to account for the observed signal.  Our approach indicates that the MCP model minimizing CV error is likely to have an mFDR no higher than 8.2\%; compared to the lasso model that minimized CV error, we can be considerably more confident in the smaller number of features selected by MCP.

MCP strikes an attractive balance in this case between prediction accuracy and false discovery rate.  One reason for this is that two genes, \textit{ZC2HC1A}, and \textit{FAM117A}, have very large effects, exactly the scenario MCP is designed to perform well in. In the lasso model with $\lambda$ selected by cross validation these variables have coefficients of -0.170 and -0.118 respectively, while in the MCP model these variables have coefficients of -0.220, -0.175 despite cross-validation selecting a larger $\lambda$ (more penalization) for MCP.  This illustrates the fact that lasso models can only accommodate large effects by lowering $\lambda$, which comes at the cost of allowing additional noise variables into the model.

For comparison we also analyzed the data using repeated sample splitting with 100 random splits, as advocated by  \citet{Meinshausen2009}. However, this approach was unable to select any genes at a false discovery rate of 10\%, or even at a more liberal rate of 50\%.
As mentioned earlier, the {\tt covTest} package does not (currently) accommodate survival data.

We also applied a large scale univariate testing approach to the Shedden data, fitting a series of Cox regression models adjusting for all clinical covariates and containing a single gene.  After using the Benjamini-Hochberg procedure to control the false discovery rate at 10\%, this approaches selects 803 genes, far more than any of the regression-based approaches.  As discussed in Section~\ref{Sec:sim-comp}, the primary difference between univariate and regression-based methods with respect to false discoveries is that univariate approaches tend to select large numbers of features that are indirectly correlated with the outcome.  Our simulation results would suggest that for these data, univariate testing yields a large number of genes, most of which are only indirectly associated with survival, while the mFDR approach yields a smaller number of genes, most of which are directly associated with survival.

\subsection{Lung cancer status among smokers}

\citet{Spira2007} collected RNA expression data for 22,215 genes from histologically normal bronchial epitheliums of $n = 192$ smokers, of which 102 had developed lung cancer and 90 had not developed lung cancer.
The goal of the study was to identify genes that are indicative of whether or not a smoker has lung cancer.  For our analysis we fit penalized logistic regression models, and compare the results using mFDR, cross validation, and the covariance test to choose $\lam$, as well as the genes selected via sample splitting.   We also compare these model-based approaches to the traditional univariate approach with false discovery rate control.

Table~\ref{Tab:case-studies} shows results for the Spira data at various $\lambda$ values for the lasso and MCP penalty functions.  In contrast with the Shedden case study, for both penalties in this example the mFDR bound is 100\% at the value of $\lam$ selected by cross validation.  In other words, we are unable to provide an upper bound for the mFDR at the value of $\lam$, and one should be cautious about interpreting the selected features as significant.  For example, at $\lam=0.029$, there may very well be some true discoveries among the 71 selected features, but a lasso-penalized logistic regression model could easily select 71 features just by chance at that value of $\lam$ even if none of them were related to the outcome.

In order to arrive at a set of features with a low false discovery rate, we much choose a much smaller model (for the lasso model, 10 features instead of 71).  As in the earlier case study, however, the price of an FDR restriction is a decrease in prediction accuracy.  The prediction error is considerably higher in the $\mFDR < 10\%$ lasso and MCP models than for the models selected by cross-validation.  In the Spira example, no genes have particularly large signals; as a result, the MCP approach is less successful than at finding a model that is attractive from both the mFDR and prediction perspectives than it was in the Shedden example.

As in the Shedden case study, neither the sample splitting nor covariance testing approaches were able to select any genes at a 10\% FDR, nor could any genes be selected by either approach at the considerably more liberal cutoff of 50\%.
Finally, we tested for significant features in a univariate approach, fitting a separate logistic regression model for each gene and controlling the false discovery rate at 10\% using the Benjamini-Hochberg procedure.
This approach identifies 2,833 genes as significant (single-gene tests based on $t$-tests yielded similar results).

Depending on the goals of the analysis and whether indirect associations are of interest, either the univariate or mFDR approach might prove useful here.
However, as in the earlier case study, fully conditional and pathwise conditional approaches to controlling the false discovery rate tend to be so restrictive that they are unable to identify any features in high-dimensional studies.

\section{Discussion}

Controlling the marginal false discovery rate of a penalized likelihood model is a useful way of assessing the reliability of a selected set of features. Unlike other approaches that have a similar goal, such as sample splitting or the covariance test, the mFDR approach uses a less strict definition of a false discovery which does not require conditioning on any other variables and consequently only limits the selection of variables that are noise in an unconditional -- i.e., marginal -- sense.  The simulations in this paper demonstrate that while mFDR is based upon a weaker false discovery definition, when used in combination with a penalized regression model it tends to perform very well at limiting the number of indirect (non-causal) feature selections.


Furthermore, mFDR is far more convenient from the standpoint of computational burden than other approaches. For the Spira data, where $n = 192$ and $p =22,215$, fitting a penalized logistic regression model and then estimating the mFDR for the entire $\lam$ sequence takes only 1.4 seconds.  This is over 400 times faster than sample splitting and the covariance test, each of which took over 9 minutes, and is nearly 30 times faster than even univariate testing, which took over 3 minutes.
Results are similar for the Shedden survival analysis.  The computational efficiency of mFDR makes it a particularly appealing tool during the early stages of an analysis when a number of candidate models are being considered.

A natural question is whether marginal independence is a realistic assumption.  In many applications, complete independence is unlikely to be literally true.  However this is, in a sense, the point: if one selects 20 features and one would expect to select 16 of them even if all of those features were completely independent of the outcome, this calls into question the confidence with which we can claim associations between any given feature and the outcome.  A relevant analogy is the testing of a point null hypothesis -- an effect of precisely zero may be unrealistic, but if that hypothesis cannot be rejected there is cause for skepticism.

The mFDR control procedure proposed here is currently implemented in the R package {\tt ncvreg} \citep{Breheny2011} using the {\tt mFDR} function. The function accepts an ncvreg fitted model and calculates the expected number of false discoveries, as well as the mFDR for each value of the $\lambda$ sequence used in fitting the model.  The package also provides a plotting method, which produces plots like those in Figure~\ref{Fig:Shedden}.  The {\tt ncvreg} package accommodates lasso, MCP, SCAD, elastic net, and MNet \citep{Huang2016} penalties and provides a convenient model summary measure for penalized linear, logistic, and Cox regression models.

In summary, this manuscript has demonstrated how to extend control over the marginal false discovery rate for penalized regression models to any general likelihood-based loss function.  Controlling the mFDR is much faster, considerably less restrictive, and is conveniently available in an open-source software package.  The marginal false discovery rate is an easily interpreted and broadly useful approach to inference concerning the reliability of selected features in a penalized regression model, capable of being generalized to a wide variety of penalty functions and modeling frameworks.

\bigskip

\noindent {{\bf Supporting information}}

\noindent Data and source code to reproduce all results and figures are available at 

\noindent\url{http://github.com/pbreheny/mfdr-likelihood-paper}.

\bibliographystyle{ims}

\end{document}